\documentclass[letterpaper, 10 pt, conference]{ieeeconf}

\IEEEoverridecommandlockouts                              

\makeatletter
\let\proof\@undefined                        
\let\endproof\@undefined                  
\makeatother 

\makeatletter
\newcommand\footnoteref[1]{\protected@xdef\@thefnmark{\ref{#1}}\@footnotemark}
\makeatother

\usepackage{cite}
\usepackage{amsmath,amssymb,amsfonts}
\usepackage{algorithmic}
\usepackage{graphicx}
\usepackage{textcomp}
\usepackage{amsmath,bm}
\usepackage{ wasysym }

\allowdisplaybreaks

\DeclareMathOperator*{\argmin}{arg\,min}

\usepackage[font=singlespacing]{caption}
\usepackage{enumerate}
\usepackage{cuted}
\usepackage{stfloats}
\usepackage{amsmath} 
\usepackage{amssymb}  
\usepackage{amsthm}
\usepackage{multirow}
\usepackage{bbm, dsfont}
\usepackage{bm} 
\usepackage{url}
\usepackage{balance}
\usepackage{color}
\graphicspath{{./figures/}}
\usepackage{array}
\usepackage{subfig}
\usepackage[a]{esvect}
\usepackage{mathtools}
\usepackage{scalerel}

\newcommand{\sgn}{\textnormal{sgn}}
\newcommand{\diag}{\textnormal{diag}}
\makeatletter
\newif\if@restonecol
\makeatother

\usepackage[linesnumbered,ruled,vlined,onelanguage]{algorithm2e}

 \newtheorem{lem}{Lemma}
\newtheorem{theorem}{Theorem}
\newtheorem{definition}{Definition}
\newtheorem{assumption}{Assumption}

\newtheorem{proposition}{Proposition}

\newtheorem{problem}{Problem}

\usepackage[font=small]{caption}
\usepackage{subcaption}
\usepackage{mathtools}

\newcommand{\sh}[1]{{\color{black} #1}}
\newcommand{\shn}[1]{{\color{black} #1}}
\newcommand{\yo}[1]{{\color{black} #1}}

\newcommand{\sy}[1]{{\color{black} #1}}
\newcommand{\pt}[1]{{\color{black} #1}}
\newcommand{\yy}[1]{{\color{black} #1}}

\title{Control Barrier Functions for Linear Continuous-Time Input-Delay\\\ Systems with Limited-Horizon  Previewable Disturbances}

\author{Tarun Pati$^*$, Seunghoon Hwang$^*$ and Sze Zheng Yong
\thanks{$^*$Authors contributed equally. T. Pati  and S.Z. Yong are with the Mechanical and Industrial Engineering Department, Northeastern  University, Boston, MA 
(
{\tt\footnotesize \{pati.ta,s.yong\}@northeastern.edu}); S. Hwang is with the School for Engineering of Matter, Transport and Energy, Arizona State University, Tempe, AZ (
{\footnotesize\tt{shwang45@asu.edu}}). 
This work was supported in part by NSF grants CMMI-1925110, CNS-2312007 and CNS-2313814.} 
}

\allowdisplaybreaks

\begin{document}

\maketitle
\thispagestyle{empty} 

\begin{abstract} 
Cyber-physical and autonomous systems are often equipped with mechanisms that provide predictions/projections of future disturbances, e.g., road curvatures, commonly referred to as preview or lookahead, but this preview information is typically not leveraged in the context of deriving control barrier functions  (CBFs) for safety. This paper proposes a novel limited preview control barrier function (LPrev-CBF) that avoids both ends of the spectrum, where on one end, the standard CBF approach treats the (previewable) disturbances simply as worst-case adversarial signals and on the other end, a recent Prev-CBF approach assumes that the disturbances are previewable and known for the entire future. 
Moreover, our approach applies to input-delay systems and has recursive feasibility guarantees since we 
explicitly take input constraints/bounds into consideration. 
Thus, our approach provides strong safety guarantees in a less conservative manner than standard CBF approaches while considering a more realistic setting with limited preview and input delays. 
\end{abstract}


\section{Introduction}

Many cyber-physical and autonomous systems (e.g., self-driving cars and robot exoskeletons)  possess forecasting tools such as forward-looking sensors (e.g., cameras and LIDAR, topographical maps), and predictive models to anticipate what lies ahead. This preview information for a window into the future, if used correctly, can significantly improve system performance and safety. However, many control strategies often neglect to use this information and  opt instead to consider worst-case scenarios, especially in the context of safety when computing robust controlled invariant sets, e.g., \cite{rungger2017computing} or control barrier functions (CBFs), e.g.,  \cite{ames2016control,ames2019control,pati2023robust}. On the other hand, optimal and model predictive control (MPC) methods, e.g.,  \cite{tomizuka1975optimal,garcia1989model}, 
do use certain preview information but often lack recursive safety and feasibility properties.

Recent research underscores the advantages of using preview information in safety controls for discrete-time systems, including those with input delays \cite{liu2021value,liu2020scalable}, 
leading to enhanced safety with increased preview time. For continuous-time systems,  
the predictive CBF method  \cite{breeden2022predictive} proposed an approach that can leverage information about `controllable' predicted trajectories, while our recent work in \cite{pati2023preview} introduced a preview CBF approach that can utilize information of previewable (but `uncontrollable') disturbances such as road gradients or curvatures or predicted future motion of other agents. However, the latter assumes that the preview horizon is as long as needed (unlimited), which is not always realistic, e.g., due to limited sensing range of cameras or LIDARs. 
Further, to  our best knowledge, such techniques for continuous-time systems with input delay, e.g., due to network latency or hardware constraints, are still lacking. CBF tools have been explored for input-delay systems   \cite{jankovic2018control,singletary2020control}, but they do not or cannot make use of preview information.

\emph{Contributions.} Building on our \yo{prior design of Prev-CBFs with unlimited preview} in \cite{pati2023preview}, this paper presents \emph{limited preview control barrier functions} (LPrev-CBFs) for linear continuous-time input-delay systems where the preview horizon for the previewable disturbances is limited and fixed, which better reflects real-world settings where sensing ranges are limited. This research relaxes the restrictive setting in \cite{pati2023preview} that assumes unlimited preview and also addresses 
safety concerns stemming from input delays. 

In contrast to standard CBFs with or without input delays, e.g., \cite{ames2016control,ames2019control,jankovic2018control}, that (implicitly) enforce robust
safety by considering the worst-case future disturbances, 
our approach leverages preview information, such as sensor data, as previewable disturbances to mitigate the conservatism. 
On the other hand, in contrast to our prior Prev-CBF approach \cite{pati2023preview} that assumes all future disturbances are previewable, our LPrev-CBF only uses the preview information of the disturbances for a limited and fixed horizon and considers the worst-case  unpreviewed disturbances beyond that horizon. 
This is particularly beneficial in minimizing the need for interventions while still ensuring safety and robustness from disturbances, 
In other words, the LPrev-CBF is designed for a realistic scenario while taking advantage of available preview information to be more permissive and the associated safety controller is less likely to be activated (i.e., less interventions when used as a safety filter) when compared to standard CBFs.   

Additionally, our closed-form LPrev-CBF explicitly incorporates input constraints into its design and consequently, it is naturally guaranteed to be recursively feasible (and safe) when incorporated into an optimization-based safety controller to modify any nominal (input-delay) controller. Further,  the results in this paper are
in itself a novel contribution even in the absence of
input delay, as an extension of our prior work in \cite{pati2023preview} to consider the setting when the preview horizon is
limited and fixed.

\sh{The efficacy of the proposed LPrev-CBFs is demonstrated via  practical 
examples
---an assistive shoulder robot equipped with interaction torque preview capabilities and a vehicle lane-keeping system that utilizes road curvature preview.} 

\pt{This paper is structured as follows. The problem of interest is outlined in Section II. Next, In Section III, we elaborate on the our proposed LPrev-CBFs to solve this problem. Then, we illustrate the efficacy of our approach in Section IV using examples of an assistive shoulder robot and vehicle lane keeping, and also 
discuss the advantages of preview for safety of a linear continuous-time  input-delay system. Finally, we conclude by presenting a summary of our contributions and some future work in Section V.}

\section{Preliminaries and Problem Formulation}
\subsection{Notations} 
$\mathbb{R}_+$ and $\mathbb{R}^n$ denote the set of non-negative real numbers and $n$-dimensional Euclidean space, respectively. An identity matrix of size $n$ is represented by $I_n$ and a $m\times n$ matrix of zeros is represented by $\mathbf{0}_{m\times n}$. Additionally, all vector inequalities represent element-wise inequalities, while $|\cdot|$ and $\sgn(\cdot)$ serve as element-wise absolute value and signum operators, respectively, and $\diag(v)$ represents a diagonal matrix whose diagonal elements are elements of a vector $v$. Further, a class $\mathcal{K}_{\infty}$ function $\alpha: [0,\infty) \rightarrow [0,\infty)$ is strictly increasing and continuous with $\alpha(0)=0$ and $\lim_{r\rightarrow\infty}\alpha(r)=\infty$.
\subsection{Problem Formulation}\label{probform}
\sh{Consider the continuous-time linear control system that includes a time-delayed control input, along with additive previewable disturbances. 
This system is denoted as $\Sigma_{delay}$ and is described as follows:}
\begin{align} \label{eqn:delay_sys} 
\hspace{-0.1cm}\begin{array}{ll}
\Sigma_{delay}:
    \dot{x}(t) \hspace{-0.05cm}=\hspace{-0.05cm} A {x}(t)\hspace{-0.07cm}+\hspace{-0.07cm} Bu(t\hspace{-0.07cm}-\hspace{-0.07cm}T_i) \hspace{-0.07cm}+\hspace{-0.07cm} B_d d(t),
\end{array}\hspace{-0.35cm}
\end{align}
where $x(t) \in \mathcal{X} \subseteq \mathbb{R}^{n}$ represents the system states, $u(t) \in \mathcal{U} \subseteq \mathbb{R}^{m}$ is the control input subject to a \yo{(fixed)} time delay $T_i$, $d(t) \in \mathcal{D} \subseteq \mathbb{R}^{l}$ denotes bounded and previewable disturbances \yo{(only for a \emph{fixed} preview horizon $T_p$, beyond which they are not previewed/uncertain)}. 
Specifically, the disturbance set is  $\mathcal{D} \triangleq \{ d \mid |d| \leq d_m \}$ and the control input set is $\mathcal{U} \triangleq \{ u \mid |u| \leq u_m \}$, where $d_m$ and $u_m$ represent the disturbance and actuation bounds\footnote{\label{note1}\pt{For ease of exposition, we assumed symmetric bounds. Any asymmetric bounds can be considered by taking their midpoints as known signals and  deviations from these midpoints as signals with symmetric bounds.}}, respectively. 


\yo{The term ``previewable disturbance" is used to refer to any exogenous inputs, signals, or parameters for which future values may be anticipated/known. Examples 
include a reference signal in tracking tasks or the predicted trajectories of other agents, as well as 
road conditions such as curvature, gradient, or friction coefficients that could be measured through limited-range sensing and perception modules.}

\yo{Inspired by the literature on time-delay control systems, we will also represent $\Sigma_{delay}$ using \emph{predicted states} $z(t) \triangleq x(t+T_i)$, as:}
\begin{align}\label{eqn:delay_sys_z} 
\hspace{-0.375cm}\begin{array}{ll}
\yo{\Sigma_{{pred}}\!:\!} 
\begin{cases}
\dot{z}(t) \hspace{-0.07cm}=\hspace{-0.07cm} A z(t) \hspace{-0.07cm}+\hspace{-0.07cm} B u(t) \hspace{-0.07cm}+\hspace{-0.07cm} B_d d(t \hspace{-0.07cm}+\hspace{-0.07cm} T_i), \\
y(t) \hspace{-0.07cm}=\hspace{-0.07cm} C z(t),
\end{cases}
\end{array}\hspace{-0.9cm}
\end{align}
where the initial predicted state $z(0)$ is as given in \eqref{eq:z_exact} with $u(\tau)$ for all $\tau \in [-T_{i},0]$.
The scalar output $y(t)\in \mathbb{R}$ 
represents the system variable of interest that relates to 
the safety of the system. Specifically, we consider \emph{system safety} as the satisfaction of a desired output constraint\footnoteref{note1} given as: 
\begin{align}\label{eq:outputbounds}
    |y(t)|=|Cz(t)| \le y_m, \quad \forall t\ge 0,
\end{align}
with known output bounds $y_m \in \mathbb{R}$. In this equivalent form, $z(t) \in \mathcal{Z} \subseteq \mathbb{R}^{n}$ serves as the new system states, and $y(t) \in \mathbb{R}$ is the scalar output.
\sh{We further assume that \yo{$\Sigma_{pred}$ with control input  $u(t)$ and disturbance (input) $d(t)$ as well as output $y(t)$} has a relative degree of 2, meaning $CB = 0$, $CB_d = 0$, $CAB \neq 0$ and $CAB_d \neq 0$.} 

In contrast to the assumption in our prior work \cite[Assumption 2]{pati2023preview} that the preview horizon is as long as needed \yo{(unlimited)}, this paper considers the (more realistic) setting where the preview horizon $T_p$ is limited and fixed, beyond which the \emph{unpreviewed} disturbance is unknown but bounded. 
\begin{assumption}\label{as:1}
The delay time $T_i$ and preview horizon $T_p$ are known and fixed/constant and they are such that $T_p> T_i$.
\end{assumption}
\begin{assumption}\label{as:2}
For a given time $t \in \mathbb{R}_+$ and known preview horizon $T_p$, the previewed disturbance $\mathbf{d}_p(t)\triangleq \{d(\tau)\in \mathcal{D},t \le \tau  \yo{\ <\ }
t+T_p\}$ is known. Further, beyond the preview horizon, we define \yo{the \emph{unpreviewed} disturbance} $\mathbf{d}_{np}(t)\triangleq \{d(\tau)\in \mathcal{D},t+T_p \le \tau  \yo{\ <\ }
\infty\}$, which \yo{is unknown but} bounded with known bounds.
\end{assumption}

By assuming $T_p > T_i$, the \emph{predicted state} $z(t)$ in \eqref{eqn:delay_sys_z}, \yo{i.e.,
\begin{align}\hspace{-0.1cm}
\begin{array}{rl}
    \textstyle z(t)\hspace*{-0.25cm}&=e^{AT_i} x(t)\\ &\ + \int_{0}^{T_i}\hspace{-0.15cm} e^{A(T_i-\tau)} (Bu(t\hspace{-0.05cm}-\hspace{-0.05cm}T_i\hspace{-0.05cm}+\hspace{-0.05cm}\tau) +B_d d(t\hspace{-0.05cm}+\hspace{-0.05cm}\tau))
     d\tau,
     \end{array}\hspace{-0.1cm}\label{eq:z_exact}
\end{align}
is} exactly known; thus, we can equivalently consider $\Sigma_{pred}$ in lieu of $\Sigma_{delay}$ and for simplicity, we shall also directly define the safe sets based on $z(t)$ under the assumption that the system is safe for all $ 0\le t < T_i$ such that $z(0)$ is within the controlled invariant set defined below.

\begin{definition}[Safe Sets]\label{def:safesets}
Let $ S_{z}\subseteq \mathbb{R}^{n} $ be a \emph{safe set} of $ \Sigma_{pred}$ that describes desirable/given safety constraints on the states, and let $S_{z,p} \subseteq \mathbb{R}^{n} \times \mathcal{D}^{[0,T_p)}$ 
be the \emph{$T_p$-augmented safe set} of $\Sigma_{pred}$, defined as 
\begin{align*}
  S_{z,p}\triangleq \{(z,\mathbf{d}_{p}) 
  \mid z \in S_{z}, \mathbf{d}_{p}\in \mathcal{D}^{[0,T_p)}\},
\end{align*}
where $\mathcal{D}^{[0,T_p)}$ is the set of all trajectories of $d(\tau)$ within the time interval of $[0, T_{p}]\triangleq \{\tau|0\le\tau<T_{p}$\}, defined as,
\begin{align*}
   \mathcal{D}^{[0,T_p)}\triangleq \{d(\tau), \forall \tau \in [0,T_{p}]\mid d(\tau)\in \mathcal{D}\}.
\end{align*}
\end{definition}

\begin{definition}[Controlled Invariant Set]
\label{def:inv_set}
A set $ \mathcal{C}\subseteq S_z $ is a \emph{\yo{robust} controlled invariant set} of $\Sigma_{pred}$ in a safe set $ \sy{S_{z}\subseteq \mathbb{R}^{n} }$ if for all $ {z}(0)\in \mathcal{C} $, there exists some $  {u}(t)\in \mathbb{R}^{m} $  such that 
 for all $ d(t)\in \mathcal{D} $, we have $z(t)\in \mathcal{C} \subseteq S_z$, $\forall t\ge 0$. $\mathcal{C}_{max}$ is \emph{the maximal \yo{robust} controlled invariant set} in \sy{$S_{z}$} if $\mathcal{C}_{max}$ contains all robust controlled invariant sets 
 in $S_{z}$.
	
	Further, a set $\mathcal{C}_{p} \in S_{z,p}$
 is a \emph{\yo{limited preview} controlled invariant set} of \yo{$\Sigma_{pred}$} %
 in an augmented safe set $ S_{z,p}$ if for all $ ({z}(0),\mathbf{d}_p(0))\in \mathcal{C}_{p} $, 
 there exists some $  {u}(t)\in \mathbb{R}^{m}$ 
 such that \yo{for all $ \mathbf{d}_{np}\in \mathcal{D}^{[T_p,\infty)}$, we have 
 $(z(t),\mathbf{d}_p(t)) 
 \in \mathcal{C}_{p} \subseteq S_{z,p}$}, $\forall t\ge 0$. $\mathcal{C}_{max,p}$ is \emph{the maximal \yo{limited preview}  controlled invariant set} in $\sy{S_{z,p}}$ if $\mathcal{C}_{max,p}$ contains all \yo{limited preview}  controlled invariant sets of $\Sigma_{pred}$ in  \sy{$S_{z,p}$}.
\end{definition}

Additionally, the definitions presented herein can be interpreted as the continuous-time analogs of their discrete-time counterparts delineated in \cite{liu2020scalable}. Of particular significance is the proof provided in the aforementioned study \yo{that the maximal controlled invariant sets for systems without preview is a subset of the maximal controlled
invariant sets for systems with preview} 
in a discrete-time framework, even 
\yo{in the presence of} input delays. Inspired by these findings, we postulate that preview may \yo{also} offer 
\yo{similar} advantages in continuous-time systems \yo{with} 
input delays.

In contrast to the objective of identifying the maximal \yo{limited preview} controlled invariant set for $\Sigma_{pred}$, this study shifts its focus towards the exploration of \yo{limited preview} control barrier functions with preview capabilities for $\Sigma_{pred}$. Specifically, we aim to render \shn{some} time-varying set $\mathcal{C}_{z,t} \subseteq S_z$ 
controlled invariant. To achieve this, we introduce a novel concept of a time-varying \yo{`limited preview safe set',} 
denoted as $\mathcal{C}_{z,p,t} \subseteq S_{z,p}$, which is not only  controlled invariant but also serves to imply the \shn{existence of some 
$\mathcal{C}_{z,t} \subseteq S_z$ that is controlled invariant by contruction/design}. It is worth noting that the efficacy of control barrier functions for systems without input delay, but with preview capabilities \yo{(for an ``infinite"/unlimited horizon)}, has been previously established in \yo{our prior work} \cite{pati2023preview}.

\begin{definition}[\yo{Limited Preview Safe Set}]
\label{def:RprevSafe}
    Given a predictive system with preview $\Sigma_{pred}$ (with known $\mathbf{d}_p \in \mathcal{D}^{[0,T_p)}$ and unknown $\mathbf{d}_{np} \in \mathcal{D}^{[T_p,\infty)}$), a \yo{super-level set} 
    $\mathcal{C}_{z,p,t}$ 
    defined on a 
    time-varying function $h:\mathcal{X} \times \mathcal{D}^{[0, T_p)} 
    \times \mathbb{R}_+ \rightarrow \mathbb{R}$: 
\begin{align}
 \mathcal{C}_{z,p,t}&\triangleq \{(z,\mathbf{d}_p,t) \mid h(z,\mathbf{d}_p,t)\ge 0,
    \},\hspace{-0.2cm}\label{eq:Cxpw}
\end{align}
which, in turn, is defined based on another time-varying function $h_{np}$ according to $h(z,\mathbf{d}_p,t)\triangleq \min_{\mathbf{d}_{np}\in \mathcal{D}^{[T_p,\infty)}} h_{np}(z,\mathbf{d}_p,\mathbf{d}_{np},t)$ with $\mathcal{D}^{[T_p,\infty)}$ being the set of all trajectories of $d(t)$ starting from $T_{p}$, and its boundary $\partial \,\mathcal{C}_{z,p,t}$ and interior $\text{Int}(\mathcal{C}_{z,p,t})$ 
similarly defined with the $\ge$ operator being replaced by $=$ and $>$, respectively, is a \emph{limited preview safe set} for $\Sigma_{pred}$ if $(z(t),\mathbf{d}_p(t),t)\in \mathcal{C}_{z,p,t}$ for all $t \ge 0$ implies that $z(t)\in S_z$ for all $t \ge 0$. 
\end{definition}

\shn{Note that, by design, the limited preview robust safe set in the above definition needs to be defined or chosen such that its controlled invariance implies the existence of some $\mathcal{C}_{z,t} \subseteq S_z$ that is controlled invariant.}

Then, the problem of interest 
can be formally written as:

\sh{\begin{problem}[Safety with \yo{Limited Preview}] 
\label{prob:1} 
Given \yo{an input-delay} system with preview $\Sigma_{delay}$ \yo{in \eqref{eqn:delay_sys} satisfying Assumptions \ref{as:1}--\ref{as:2}}, its corresponding \yo{equivalent predictive system $\Sigma_{pred}$ in \eqref{eqn:delay_sys_z}} and a safe set $S_z$ (cf. Definition \ref{def:safesets}), 
construct a \yo{limited preview} control barrier function \yo{(LPrev-CBF)} 
\sy{corresponding to \sy{$\mathcal{C}_{z,p,t}$} in \eqref{eq:Cxpw}} that guarantees 
 \yo{limited preview} controlled invariance of $\Sigma_{pred}$ in \sy{${S}_{z}$} \yo{(and thus,  safety of $\Sigma_{delay}$ under Assumption \ref{as:1})}.  
\end{problem}}

\vspace{-8pt}\section{Main results}\label{sec:main}

\yo{We now describe our proposed method to address Problem \ref{prob:1}, where we introduce a novel \yo{limited preview} CBF in closed-form and describe how it can be implemented computationally to guarantee safety.} 
\subsection{Limited Preview Control Barrier Functions \yo{(LPrev-CBFs)}}
\pt{First, we present the definition of  \yo{Limited Preview Control Barrier Function (LPrev-CBF) as an extension of} 
Preview Control Barrier Function (Prev-CBF) in \cite{pati2023preview}.}
\begin{definition}[\yo{Limited Preview} \sy{CBF}]
\label{def:rCBF}
\pt{Given an input-delay system $\Sigma_{delay}$, its corresponding predictive system $\Sigma_{pred}$ with a fixed-horizon previewable disturbance that satisfies Assumptions \ref{as:1}-\ref{as:2} and a safe set $S_z$ (cf. Definition \ref{def:safesets}), then \yo{a continuously differentiable function $h:\mathcal{X} \times \mathcal{D}^{[0, T_p)}  \times \mathbb{R}_+ \rightarrow \mathbb{R}$} 
is a \yo{limited preview} CBF corresponding to uncertainty dependent safe set $\mathcal{C}_{z,p,t}$ \shn{in \eqref{eq:Cxpw}}, if \yo{there exist a control input $u\in \mathcal{U}$ and} a class $\mathcal{K}_\infty$ function $\alpha$  such that:}
\begin{gather}\label{eq:rCBF}
    \dot{h}(z,u,\mathbf{d}_p,t)\hspace{-0.1cm}\geq\hspace{-0.1cm}-\alpha(\yo{h(z,\mathbf{d}_p,t)}), 
\end{gather}
for all $z \in \mathcal{X}$ and $t \in \mathbb{R}_+$.  
Further, for any $t\in\mathbb{R}_{+}$, $z \in \mathcal{X}$ and \yo{$\mathbf{d}_p\in \mathcal{D}^{[0,T_p)}$}, an associated 
safe input set is defined as: 
\begin{align}
\begin{array}{r}
\hspace{-0.18cm} K_{\mathcal{C}}(z,\mathbf{d}_p\sy{,t})\hspace{-0.08cm}=\hspace{-0.08cm}\{u\in \mathcal{U} \mid \yo{\eqref{eq:rCBF} 
\text{ holds}}\}.
\end{array}
\end{align}
\end{definition}
\vspace{-8pt}\pt{\begin{theorem}[Safety with Limited Preview]\label{rbs:safe}
Given a predictive system $\Sigma_{pred}$ with a fixed-horizon previewable disturbance that satisfies Assumptions \ref{as:1}--\ref{as:2} and a safe set $S_z$ (cf. Definition \ref{def:safesets}), if $ h $ is a LPrev-CBF corresponding to the \yo{limited preview} safety set $ \mathcal{C}_{z,p,t} $ from  \eqref{eq:Cxpw}, then for the predictive system $\Sigma_{pred}$ with $z(0)\in S_z$, any Lipschitz continuous controller $ u(x,\yo{\mathbf{d}_p},t)\in K_{\mathcal{C}}(z,\mathbf{d}_p,t)$ \yo{with $z$ in \eqref{eq:z_exact}} ensures the  controlled invariance of the \yo{limited preview}  safety set $ \mathcal{C}_{z,p,t} $. \shn{Consequently, there exists some set $ \mathcal{C}_{z,t}\subset S_z $ for the system $ \Sigma_{pred} $ for which $ u(x,\yo{\mathbf{d}_p},t) $ also ensures its controlled invariance.}
Thus, the predictive system $ \Sigma_{pred} $ with preview is guaranteed to be safe, i.e., $ z(t) \in S_z,  \forall t \geq 0 $.
\end{theorem}}
\begin{proof}
\vspace{-10pt}\pt{
If $h$ is a LPrev-CBF associated with the \yo{limited preview} safe set $\mathcal{C}_{z,p,t}$, then any controller $u(x,\yo{\mathbf{d}_p},t) \in K_{\mathcal{C}}(z,p,t)$ enforces \eqref{eq:rCBF} 
for all $z\in\mathcal{X},\ \forall t\ge0$; hence $\mathcal{C}_{z,p,t}$ as defined in \eqref{eq:Cxpw} is forward control invariant, i.e., $h(z,\mathbf{d}_p,t)\ge 0, \ \forall t\ge0$. Consequently, the predictive system $\Sigma_{pred}$ and in turn the \yo{input-delay} system $\Sigma_{delay}$ are 
safe with 
preview for all $ t\ge0$ with respect to the set $\mathcal{C}_{z,t}\subseteq S_z$, \shn{where $\mathcal{C}_{z,t}$ exists by construction.} 
}
\end{proof}

\pt{\subsection{Closed-Form Candidate \yo{Limited Preview} CBF}
We now introduce a framework to formulate a closed-form candidate \yo{limited preview} CBF (LPrev-CBF)  to derive a \yo{limited preview} 
safe set (cf. Definition \ref{def:RprevSafe}) for the predictive system $\Sigma_{pred}$, given desired output constraints in \eqref{eq:outputbounds}. 

\pt{
The proposed framework is inspired by our \yo{prior work} \cite{pati2023preview}, where a predictor-based method \cite{fridman2014introduction,jankovic2018control} is considered, with the predicted state $z$ projected $T$ seconds into the future as: 
\begin{align} \label{eq:x_T_p}
\begin{array}{rl}
    \hspace{-0.15cm}
    z(t\hspace{-0.05cm}+\hspace{-0.05cm}T)\hspace{-0.075cm}=\yo{\phi(t,T)}+\epsilon(t,T) +\hspace{-0.05cm}\int_{0}^{T}\hspace{-0.15cm} e^{A(T-\tau)} Bu(t\hspace{-0.05cm}+\hspace{-0.05cm}\tau) 
     d\tau,
\end{array}\hspace{-0.1cm}
\end{align}
where $\yo{\phi(t,T)}\triangleq \yo{e^{AT} z(t) +}\int_{0}^{T_{\delta}} e^{A(T-\tau)}B_d d(t+T_i+\tau) d\tau$,  $T_{\delta}\triangleq \min(T_p-T_i,T)$ and $\epsilon(t,T)\triangleq  \int_{T_{\delta}}^{T} e^{A(T-\tau)}B_d d(t+T_i+\tau) d\tau\hspace{-0.05cm} $ . Note that at any given $t$ and given preview of the disturbance \yo{$T_p > T_i$} seconds into the future, $\yo{\phi(t,T)}$ can always be pre-computed, \yo{while $\epsilon(t,T)$ contains unpreviewed disturbances  that the LPrev-CBF must be robust with respect to their worst-case realizations}.
Then, to guarantee the satisfaction of the output bounds in \eqref{eq:outputbounds}, $\forall t\ge 0$, we enforce that the (immediate) future minima or maxima of worst-case output trajectories under maximum acceleration or deceleration \yo{inputs}, respectively, always satisfies the output bounds, $t\ge0$. Imposing these minima or maxima constraints also ensures the constraint feasibility at all times from the current time $t$ to the time associated with the minima or maxima, which we call the worst-case (minimum) \emph{stopping time}, as defined below:}
\pt{
\begin{definition}[\yo{Worst-Case} Stopping Time]\label{def:Sotpping}
At any given time $t \in \mathbb{R}_{+}$ for the predictive system $\Sigma_{pred}$ with fixed-horizon preview in \eqref{eqn:delay_sys_z}, we define the worst-case (minimum) stopping time $T_s(t)$ as the minimum $T_s(t)$ such that the worst-case output velocity \yo{$\dot{y}_w(t + T_s(t))=C\dot{z}_w(t+T_s(t)) =CAz_w(t+T_s(t))= 0$} 
under maximum \yo{control input acceleration and  disturbance-induced deceleration} when $\dot{y}(t)=C\dot{z}(t) \le 0$ or maximum \yo{control input deceleration and  disturbance-induced acceleration} when $\dot{y}(t)=C\dot{z}(t) \ge 0$. 
\end{definition}

\yo{From \eqref{eqn:delay_sys_z}, under the relative degree 2 assumption, 
\begin{align} 
\shn{\ddot{y}(t)=CA^2 z(t)+CABu(t)+CAB_d d(t),} \label{eq:y_ddot}
\end{align}
from which we can observe that when $\dot{y}(t)\le 0$ \footnote{\pt{Note that per Assumption \ref{as:1}, $\dot{y}(t)=C\dot{z}(t)=CAz(t)$ is exactly known.}}, the maximum obtainable acceleration is with $\diag(\sgn(CAB))u_m$ under worst-case disturbance $-\diag(\sgn(CAB_d))d_m$, while when $\dot{y}(t)\ge 0$, the maximum deceleration is obtained with $-\diag(\sgn(CAB))u_m$ under worst-case disturbance $\diag(\sgn(CAB_d))d_m$. In other words, the worst-case output velocity $\dot{y}_w(t + T_s(t))$ and worst-case output $y_w(t+T_s(t)$ can be found by applying $u(\tau)=\hat{u}(t)$ with
\begin{align} 
\hat{u}(t)\triangleq -\sgn(\dot{y}(t))\diag(\sgn(CAB))u_m \label{eq:uhat}
\end{align}
under worst-case disturbance $d(\tau)=\hat{d}(t)$ with 
\begin{align} 
\hat{d}(t)\triangleq sgn(\dot{y}(t))\diag(\sgn(CAB_d))d_m
\label{eq:dhat}
\end{align}
for all $\tau \in [t,T_s(t)]$, resulting in, 

\vspace{-0.3cm}{\small$$ \ddot{y}_w(\tau)=CA z_w(\tau) - sgn(\dot{y}(t))|CAB| u_m +sgn(\dot{y}(t))|CAB_d| d_m,$$}\noindent\shn{which is obtained from \eqref{eq:y_ddot} with $u(t)$ in \eqref{eq:uhat} and $d(t)$ in \eqref{eq:dhat},}
and the worst-case stopping time is the $T_s(t)$ that is the solution to $\dot{y}_w(t + T_s(t))=CAz_w(t+T_s(t))=0$.} 
}

\sh{
It is worth noting that for a known/computed time-varying \yo{worst-case} stopping time $T_s(t)$ at given time $t\in \mathbb{R}_+$ but under Assumption \ref{as:1} with a fixed preview horizon $T_p$ and fixed input delay $T_i$, Assumption \ref{as:2} when applied to the predicted system $\Sigma_{pred}$ translates to two distinct cases: (1) \yo{When $T_p-T_i > T_s(t)$ (i.e., when the constant time horizon $T_p$ exceeds the stopping time for $z(t)$), the disturbances $d(t+T_i)$ in \eqref{eqn:delay_sys_z} are previewed/known for the entire time interval up to $t+T_s(t)+T_i$}, and 
(2) when \yo{$T_p -T_i \le T_s(t)$, the previewable disturbances $d(t+T_i)$ within the time interval $t+T_p\le \tau\le  t + T_s(t)+T_i$} 
remain indeterminate\yo{/unpreviewed but bounded by $\mathcal{D}$}.  
}

}
\pt{
\yo{Further, n}ote that the idea of worst-case (minimum) stopping time is akin to \yo{and an extension of} the (minimum) stopping time in \cite{pati2023preview}, \yo{while} the concept of utilizing the immediate future minima or maxima of the output trajectory is inspired by the idea of minimizers/maximizers in \cite{breeden2022predictive,pati2023preview}  for a ``predicted" desired/reference trajectory. \yo{In particular, 
by enforcing that the worst-case predicted outputs $T_s(t)$ seconds into the future, i.e., \begin{align}\label{eq:futureoutputbounds}
    |Cz_{w}(t+T_s(t))|\le y_{m}, \quad \forall t\ge 0,
\end{align}
with $Cz_{w}(t+T_s(t))$ corresponding to a minimizer/maximizer, we are essentially ensuring the satisfaction of the output constraints for a future moving time horizon that includes the current time. Hence, ensuring the robust controlled invariance of \eqref{eq:futureoutputbounds} corresponding to $\mathcal{C}_{z,p,t}$ implies the robust controlled invariance of \eqref{eq:outputbounds} corresponding to $\mathcal{S}$. }
}

\pt{Next, we present a closed-form candidate LPrev-CBF to compute a controlled invariant \yo{limited preview}  safe set. Note that for the remainder of this manuscript, the (explicit) dependence of $T_s$ and other terms on the current time $t$ is omitted for brevity.}
\pt{\begin{lem}[Closed-Form Candidate \yo{Limited Preview CBF}]\label{first}
Suppose Assumptions \ref{as:1}--\ref{as:2} hold. Then,
\begin{align} \label{eq:h}
\begin{array}{rl}
    \yo{h(z,\mathbf{d}_p,t)}\hspace*{-0.2cm}&=\hspace{-0.0cm}y_m\hspace{-0.cm}-\hspace{-0.0cm}\sgn{(\dot{y}(t))}y_{w}(t+T_s) \\
    &
    \hspace{-0.cm}=\hspace{-0.0cm}y_m\hspace{-0.cm}-\hspace{-0.0cm}\sgn{(\dot{y}(t))}Cz_{w}(t+T_s) \\
    &
    \hspace{-0.cm}\ge 0,
\end{array}
\end{align}
with $z_{w}(t+T_s)=\yo{\phi(t,T_s)+\hat{\epsilon}(t,T_s)}\hspace{-0.05cm}+(\int_{0}^{T_s} e^{A(T_s-\tau)}d\tau)B\hat{u}$ being the worst-case predicted state derived from \eqref{eq:x_T_p}, with $\phi(t,T_s)$ as defined below \eqref{eq:x_T_p} (when $T=T_s$) and $\hat{\epsilon}(t,T_s)\triangleq 
(\int_{T_{\delta}}^{T_s} e^{A(T_s-\tau)}d\tau)B_d \hat{d}\hspace{-0.05cm}$ with \yo{$\hat{d}(\tau)$ defined in \eqref{eq:dhat} 
and $T_{\delta}=\min(T_p-T_i,T_s(t))$, as well as 
$\hat{u}$ defined in \eqref{eq:uhat},} 
is a valid candidate LPrev-CBF that guarantee the satisfaction of the safety bounds in \eqref{eq:futureoutputbounds} corresponding to $\sy{\mathcal{C}}_{z,p,t}$ and, consequently, the safety bounds in \eqref{eq:outputbounds}
associated with $S_{z}$. 
\end{lem}}
\vspace{-10pt}\pt{\begin{proof}
First, when $\dot{y}(t)\le0$ under maximum acceleration input $\hat{u}$, the desired worst-case safety constraint is 
\yo{$y_{w}(t+T_s(t))\geq{-y_m}$, where $y_w(t+T_s(t))$ is the smallest possible  $y$  when the system comes to a stop (before changing directions) under the worst-case unpreviewed disturbance $\hat{d}$. Similarly, when $\dot{y}_w(t)\ge0$} 
under maximum deceleration input $\hat{u}$, the desired worst-case safety constraint is \yo{$y_{w}(t+T_s)\leq{y_m}$.} 
Combining these two constraints yields 
\begin{gather} \label{eq:h_1}
y_m\hspace{-0.cm}-\hspace{-0.0cm}\sgn{(\dot{y}(t))}y_{w}(t+T_s) \geq{0}.
\end{gather}
Further, $y_{w}(t+T_s)=Cz_{w}(t+T_s)$, where as described \yo{above \eqref{eq:uhat},} 
$z_{w}(t+T_s)=C y_w(t+T_s)$ is the worst-case predictive output that can be derived from  \eqref{eq:x_T_p} by considering $T=T_s$ as well as $u(\tau)=\hat{u}(t)$ and $d(\tau)=\hat{d}(\tau), \ \forall \tau\in[t+T_p,t+T_s]$ with $\hat{u}$ and $\hat{d}$ as defined in \eqref{eq:uhat} and \eqref{eq:dhat}, 
respectively.

Thus, enforcing \eqref{eq:h_1} in turn enforces \eqref{eq:futureoutputbounds}.  Hence, $h$ is a valid candidate \yo{limited preview} control barrier function (LPrev-CBF), i.e., there always exists a piece-wise constant input $u(\tau)=\hat{u}(t)= -\sgn(\dot{y}(t))\diag(\sgn(CAB))u_{m}$, for all $\tau \in [t,t+T_s]$ that enforces $h(z,\mathbf{d}_p,t)\ge 0$ for all $t\ge 0$. Consequently, due to the guaranteed feasibility of \eqref{eq:futureoutputbounds} at maxima or minima (i.e., when $\dot{y}_{w}(t+T_s)=C\dot{z}_{w}(t+T_s)=0$), the safety constraints is also feasible for the whole time horizon from $t+T_i$ to $t+T_i+T_s$; hence, \eqref{eq:outputbounds} holds.
\end{proof}}

\vspace{-10pt}\subsection{\yo{Worst-Case} Stopping Time}\label{sec:T_s}
\pt{As seen from \eqref{eq:h} in Lemma \ref{first}, the design of the candidate LPrev-CBF depends on the worst-case stopping time $T_s$.}
\begin{lem}[\yo{Worst-Case} Stopping Time] \label{lem:T_s}
 \pt{At any given time $t$ the worst-case stopping time is the smallest positive solution \yo{$T_s(t)$} to the equality $CAz_w(t+T_s)=0$ with $z_w(t+T_s)$ given below  \eqref{eq:h}, \yo{i.e., 
\begin{align}\label{eqn:stop_exp}
CA\phi(t,T_s)+ Ce^{AT_s}B\hat{u}(t)+ Ce^{A(T_s-T_{\delta})}B_d\hat{d}(t) =0,
\end{align}
where} $\hat{u}(t)$, $\hat{d}(t)$ are defined in \eqref{eq:uhat} and 
\eqref{eq:dhat}, 
respectively, \yo{$T_{\delta}$ and $\phi(t,T_s)$} are defined below  \eqref{eq:x_T_p} (with $T=T_s$).

 }
\end{lem}
\begin{proof}
   
    \pt{
    At any given time $t$ as per the predictive state dynamics in \eqref{eqn:delay_sys_z}, the output velocity $T_s$ seconds into the future is  given as $\dot{y}(t+T_s)=C\dot{z}(t+T_s)=CAz(t+T_s)$, (since the system has relative degree 2 with respect to both the input $u(\cdot)$ and the disturbance $d(\cdot)$, $CB=CB_d=0$). Consequently, with $z_w(t+T_s)$ below \eqref{eq:h}, the worst-case output velocity with \yo{$\hat{u}(t)$ and $\hat{d}(t)$ for all $\tau \in [t,t+T_s(t)]$ is given by
    \begin{align*}
    \begin{array}{l}
        \dot{y}_w(t+T_s)= CAz_w(t+T_s)\\
        \textstyle=CA \phi(t,T_s)+C\int_{0}^{T_s} Ae^{A(T_s-\tau)}d\tau B\hat{u}\\
         \textstyle\quad +C\int_{T_{\delta}}^{T_s} Ae^{A(T_s-\tau)}B_d \hat{d}(\tau) d\tau\\
        =CA \phi(t,T_s)+C(e^{A T_s}-I)B \hat{u}+C(e^{A (T_s-T_\delta)}-I)B_d \hat{d}\\
        =CA \phi(t,T_s)+Ce^{A T_s}B \hat{u}+Ce^{A (T_s-T_\delta)}B_d \hat{d},
        \end{array}
    \end{align*}
    where the final equality holds since $CB=0$ and $CB_d=0$ by the relative degree 2  assumption.}
    }
\end{proof}

\subsection{Closed-Form \yo{Limited} Preview Control Barrier Function}
\pt{Now that we have a candidate LPrev-CBF from Lemma \ref{first} and an expression for the \yo{worst-case} stopping time in Lemma  \ref{lem:T_s}, we can prove \yo{that} the candidate LPrev-CBF 
\yo{satisfies} 
the definition of limited preview CBF in Definition \ref{def:rCBF}.}
\pt{
\begin{proposition}[Closed-Form LPrev-CBF]\label{prop:Prev-CBF} Given a \yo{input-delay} system with preview $\Sigma_{delay}$ and a corresponding predictive system $\Sigma_{pred}$ that 
satisfies Assumptions \ref{as:1}--\ref{as:2} with \yo{worst-case} stopping time $T_s(t)$ computed based on Lemma \ref{lem:T_s}, the continuously mapping 
\yo{$h: \mathbb{R}^n \times \mathcal{D}^{[0,T_p)}\times \mathbb{R}_+ \to \mathbb{R}$}  in Lemma \ref{first} 
is a limited preview control barrier function (LPrev-CBF) for $\Sigma_{pred}$, if \yo{\yo{there exist a control input $u\in \mathcal{U}$ and} a class $\mathcal{K}_\infty$ function $\alpha$} 
that satisfy \eqref{eq:rCBF} 
with
\pt{
\begin{align}   \label{eq:cbf_dot}
\hspace{-0.3cm}\begin{array}{r}
\dot{h}(z,u,\mathbf{d}_p\sy{,t})
= -\sgn{(\dot{y}(t))}[Ce^{AT_s}(Az(t)+Bu(t) \\
+B_dd(t+T_i))
+\psi(t,T_s)],
\end{array}\hspace{-0.3cm}
\end{align}
with $\psi(t,T_s) \triangleq \int_{0}^{T_{\delta}} Ce^{A(T_s-\tau)}B_d \dot{d}(t+T_i+\tau) d\tau$, where $T_{\delta}$ is defined below  \eqref{eq:x_T_p} (with $T=T_s$).
}
\pt{
Further, \eqref{eq:futureoutputbounds} holds and consequently, the output constraint in  \eqref{eq:outputbounds} holds.
} 

\end{proposition}}
\begin{proof} 
\vspace{-10pt} We begin the proof by considering the closed-form candidate LPrev-CBF $h$ from  \eqref{eq:h} in Lemma \ref{first}. Next, by applying Theorem \ref{rbs:safe} to $h$, a closed form expression for $\dot{h}$ in \eqref{eq:rCBF} is obtained by computing the derivative of $h$ with respective to current time $t$. Consequently, $\dot{h}(z,u,\mathbf{d}_p\sy{,t})=-\sgn{(\dot{y}(t))}\frac{d}{dt}y_w(t+T_s)$, where $\frac{d}{dt}y_w(t+T_s)$ is the time derivative of $y_w(t+T_s)=Cz_w(t+T_s)$ with $z_w(t+T_s)$ defined below \eqref{eq:h}, which can be derived by employing Leibnitz integral rule and leveraging the fact that $CB=0$ and $CB_d=0$ (relative degree 2 assumption) to obtain

\vspace{-5pt}\begin{align} \label{eq:dydt1}
\hspace{-0.3cm}\begin{array}{l}
\frac{d}{dt}y_{w}(t+T_s)=Ce^{AT_s}\dot{z}(t)+\psi(t,T_s)\\
\hspace{0.2cm}+CA(\phi(t,T_s)+\hat{\epsilon}(t,T_s)+(\int_0^{T_s}e^{A(T_s-\tau)}d\tau)B\hat{u})\dot{T}_s\\
\hspace{0.2cm}+Ce^{A(T_s-T_{\delta})}B_d (d(t+T_i+T_{\delta})+\hat{d})\dot{T}_{\delta}\\
= Ce^{AT_s}\dot{z}(t)+\psi(t,T_s)+CAz_w(t+T_s)\dot{T}_s\\
\hspace{0.2cm}+Ce^{A(T_s-T_{\delta})}B_d (d(t+T_i+T_{\delta})+\hat{d})\dot{T}_{\delta}, 
\end{array}\hspace{-0.3cm}
\end{align}
where we defined $\psi(t,T_s)$ below  \eqref{eq:cbf_dot}, with $\phi(t,T_s)$ and $\hat{\epsilon}(t,T_s)$ defined below \eqref{eq:x_T_p} 
and \eqref{eq:h}, respectively, and applied  the definition of $z_w(t+T_s)$  below \eqref{eq:h} in the second equality.

Next, by Lemma \ref{lem:T_s}, $CAz_w(t+T_s)=0$, i.e., the third term in the above becomes $0$. Additionally, since $T_{\delta}=\min(T_p-T_i,T_s)$, we have that when $T_{\delta}=T_p-T_i$, $\dot{T}_{\delta}=0$ ($T_p$ and $T_i$ are fixed constants) and when $T_{\delta}=T_s$,  $Ce^{A(T_s-T_{\delta})}B_d=CB_d=0$ (by relative degree 2 assumption); consequently, the final term in the above  that multiplies $\dot{T}_{\delta}$ is also equal to $0$. Thus, the expression for $\dot{h}$ simplifies to
$\dot{h}(z,u,\mathbf{d}_p\sy{,t})=-\sgn{(\dot{y}(t))}\frac{d}{dt}y_{w}(t+T_s)\hspace{-0.0cm}=\hspace{-0.05cm}-\sgn{(\dot{y}(t))}(Ce^{AT_s}\dot{z}(t)+ \psi(t,T_s))$. Finally, we obtain \eqref{eq:cbf_dot} by substituting the expression for $\dot{z}(t)$ from the predictive state dynamics in \eqref{eqn:delay_sys_z}.
\renewcommand{\qedsymbol}{}
\end{proof}

\vspace{-10pt}\subsection{Optimization-Based Safety Control}
\pt{Next, the proposed LPrev-CBF is coupled with 
\yo{a nominal} controller to minimally modify it to guarantee safety.}
\begin{proposition}[Optimization-Based Safety Control]\label{prop:control}
\pt{For the \yo{input-delay} system $\Sigma_{delay}$ in \eqref{eqn:delay_sys}, at any time $t$ any (stabilizing) 
\yo{nominal (input-delay) controller $u = k(x,z,t)$ with $z(t)$ in \eqref{eq:z_exact}, if needed,} 
 can be minimally modified to guarantee safety by computing a new safe control input $u(x,\mathbf{d}_p,t)$ that is a solution to the following quadratic program (QP):}
\begin{align}\label{eq:legacy_qp}
\begin{array}{l}
 \sy{u(x,\mathbf{d}_p,t})=   \textstyle\argmin\limits_
{u \in \mathcal{U}} \frac{1}{2} \|u-k(x\sy{,z,t})\| \\
\hspace{1.8cm}s.t. \ P(t) u \le q(t),
\end{array}
\end{align}
\pt{with \yo{$z$,} $\psi(t,T_s)$, $h(z,\mathbf{d}_p,t)$ and $T_s$ from \yo{\eqref{eq:z_exact}}, Proposition \ref{prop:Prev-CBF} (as defined below \eqref{eq:cbf_dot}), Lemma \ref{first} and Lemma \ref{lem:T_s}, respectively, and a class $\mathcal{K}_{\infty}$ function $\alpha$, such that:}
\pt{\begin{gather*}
    \begin{array}{rl}
    P(t) \triangleq & \hspace{-0.3cm}\sgn(\dot{y}(t)) C e^{A T_s(t)}B,\\
    q(t) \triangleq & \hspace{-0.3cm} \alpha(h(z,\mathbf{d}_p,t))
     - \sgn(\dot{y}(t))(\psi(t,T_s) \\
     &\
     \hspace{-0.45cm}
     +Ce^{A T_s(t)}(Az(t)+B_dd(t+T_i))).
    \end{array}
\end{gather*}}
\end{proposition}

\vspace{-15pt}\begin{proof} \pt{The LPrev-CBF constraint in Theorem \ref{rbs:safe} and Definition \ref{def:rCBF}, $\dot{h}(z,u,\mathbf{d}_p,t)\ge -\alpha(h(z,\mathbf{d}_p,t))$, with $h(z,\mathbf{d}_p,t)$ from Lemma \ref{first} and $\dot{h}(z,u,\mathbf{d}_p,t)$ from Proposition \ref{prop:Prev-CBF}   can be written as:
   \begin{align*}
    \begin{array}{rl}
    -\alpha(h(z,\mathbf{d}_p,t))&\hspace{-0.3cm}\le -\sgn{(\dot{y}(t))}(Ce^{AT_s}(Az(t)+Bu(t)\\
&\
B_dd(t+T_i))+\psi(t,T_s)),
    \end{array}
\end{align*}
which can be rearranged in the form of $P(t)u\le q(t)$ with the $P(t)$ and $q(t)$ given above. 
}
\end{proof}

\vspace{-8pt}\yo{In the above, solving \eqref{eqn:stop_exp} analytically to find $T_s(t)$ for the application of Proposition \ref{prop:control} is non-trivial, but it can be found numerically, e.g., using  MATLAB functions \texttt{fsolve}, \texttt{fzero} or \texttt{vpasolve}.} 
\yo{Further, note that in the absence of input delay (i.e., when $T_i=0)$, the results in this paper are in itself a novel contribution for when the preview horizon is limited and fixed, in contrast to our prior work in \cite{pati2023preview} that assumed unlimited preview.}
\section{Illustrative Examples}
\subsection{Assistive Shoulder Exoskeleton Robot}\label{sec:simulation}

\sh{Consider the dynamics of an industrial shoulder exoskeleton robot system \cite{ott2010unified} with fixed input delay  given by:
\begin{align}\label{eq:admittance}
I_{j}\Ddot{e}(t)+B_{j}\Dot{e}(t)+K_{j}e(t) = \tau_{e}(t) +  u(t-T_i),
\end{align}
where $e(t)=\theta(t)-\theta_d(t)$ is the angular displacement error, with $\dot{e}$ and $\ddot{e}$ as its velocity and acceleration. The terms $I_{j}= I_{h}+I_{r}$, $B_{j}= B_{h}+B_{r}$, and $K_{j}= K_{h}+K_{r}$ represent the combined inertia, damping, and stiffness of the human-exoskeleton system. Here, subscripts $h$ and $r$ denote human and robot components, respectively.}

\sh{Assuming the robot-human shoulder joint is aligned and the interaction torque $\tau_{e}$ is previewable, satisfying Assumption \ref{as:2}, the control input $u$ functions as either a spring or damping force to keep system states within safety boundaries, specifically ensuring $|e(t)|\le \delta$.}
\sh{When the human-robot system \eqref{eq:admittance} is transformed into the state space form of the predicted system\footnote{\label{note3}\yo{Note that we directly use the predicted system since $z(t)$ is exactly known under Assumption \ref{as:2} and also such that we can compare with other related approaches in the literature that do not consider input delays. Moreover, using this system for both examples allows us to illustrate the benefits of preview even in the absence of input delays.}} as described in \eqref{eqn:delay_sys_z}, its state is represented by $z(t) = \begin{bmatrix} e^{T}(t+T_i)  &\dot{e}^{T}(t+T_i)\end{bmatrix}^{T}$. The previewable disturbance for this system is given by $d(t+T_i)=\tau_e(t+T_i)$ and the matrices for this system are:
\begin{gather*}
    A = \begin{bmatrix} 0 & 1 \\ -I_{j}^{-1}K_{j} & -I_{j}^{-1}B_{j} \end{bmatrix}, B =  \begin{bmatrix} 0 \\ I_j^{-1} \end{bmatrix},
    B_d = \begin{bmatrix} 0 \\ I_{j}^{-1} \end{bmatrix},\\ 
    B_w = \begin{bmatrix} 0 \\ I_{j}^{-1} \end{bmatrix},
    C = \begin{bmatrix} 1 & 0\end{bmatrix}.
\end{gather*}
The output $y(t)=e(t+T_i)$ must satisfy $|e(t+T_i)|\le \delta$, and the control input is bounded as $|u(t)|\le u_m$.}

\sh{In this study, the simulation parameters are: $I_{j}=1\, Nms^{2}/rad$, $B_{j}=2 \,Nms/rad$, $K_{j}=2\,Nm/rad$, $\tau_{e}(t)=0.43\sin(0.2\pi t) \, Nm$, and $\delta=0.2 \,rad$. The constant or fixed preview horizon and input delay times are set at $T_p=10  \,ms$ and $T_i=8 \,ms$, respectively. Input delays in robotic systems are typically in the order of milliseconds  \cite{andersen2015measuring}, but since our shoulder robot system in \cite{hunt2018new} is not commercial and operates at a sample rate of $4\, ms$, for this simulation, we have chosen a fixed time delay of $8\, ms$, which is double the sample rate of our system.
The control input $u_{m}$ varies between 1.119 to 2.0 to analyze its effect on proposed LPrev-CBF approach in Section \ref{sec:main} and contrasted with the standard CBF approach in \cite{ames2016control}, and our prior Prev-CBF \cite{pati2023preview}, which does not impose limitations on the preview horizon. For simplicity, the nominal controller is zero, meaning that $u(t)$ also indicates the safety controller's intervention. The rationale behind selecting $u_{m}$ from 1.119 is to illustrate an instance in which 
the standard CBF fails to be safe.}

\subsubsection{\yo{LPrev-CBF}s} \sh{The LPrev-CBF, delineated in Section \ref{sec:main}, ensures safety and controlled invariance for the shoulder robot, requiring $|y(t)|=|C z(t) | \le y_m = \delta$. The control input $u(t)$ adheres to the constraint $|u(t)|\le u_m$. We specifically employ the optimization-based safety controller from Proposition \ref{prop:control}, with the nominal control input set to zero, denoted as $k(z,t)=0$. The closed-form LPrev-CBF is further elaborated in Lemmas \ref{first} and \ref{lem:T_s}, and Proposition \ref{prop:Prev-CBF}.}

\subsubsection{Standard CBFs} 
\sh{We then compare our approach with the standard CBF  from \cite{ames2016control}, specifically from the lane keeping example in \cite[Section V-B]{ames2016control}. Besides the output constraint:
$$|y(t)|=|e(t+T_i)| =|C z(t)| \le y_m,$$
it assumes a bounded output acceleration:
$$|\ddot{y}(t)| =|\ddot{e}(t+T_i)| = |C\ddot{z}(t)| \le a_{\max},$$
which, given input constraints $|u|\le u_m$, is inherently bounded by the dynamics in \eqref{eq:admittance}:
$$u=I_d \ddot{e}+B_{d}\Dot{e}+ K_{d}e - \tau_{e}.$$
\yy{Then, using worst-case bounds on $\ddot{e}$, $\dot{e}$, $e$, and $\tau_e$ given by $\ddot{e}_{\max}=a_{\max}$, $\dot{e}_{\max}$, ${e}_{\max}=y_m$, and $\tau_{e,\max}$ (from Assumption \ref{as:2}), respectively,} 
the triangle inequality gives:
$$|u| \le I_d a_{\max} +B_d \dot{e}_{\max} + K_d e_{\max} + \tau_{e,\max} \yy{\ \triangleq\ } u_m.$$
The output constant acceleration bound $a_{\max}$ is then:
\begin{gather} \label{eq:amax}
    a_{\max}= I_d^{-1}(u_m - B_d \dot{e}_{\max} - K_d e_{\max} - \tau_{e,\max}).
\end{gather}
To meet this bound, the standard CBF method's control input must be:
\begin{gather}\label{eq:input-bound-amax}
    u (t) \in [-I_d a_{\max}+F_0(t),I_d a_{\max} +F_0(t)],
\end{gather} 
with $F_0(t)\triangleq B_d \dot{e}(t) +K_d e(t) - \tau_e (t)$.}

\sh{The standard lane keeping CBF approach from \cite[Section V-B]{ames2016control} proposes the following CBF:
\begin{align}\label{eq:CBFs_ames}
h(z) = (y_{m}-sgn(\dot{y}(t))y(t))-\frac{\dot{y}(t)^{2}}{2 a_{\max}}.
\end{align}
In our simulation, parameters are set as: $\dot{e}_{\max}=0.1326$, $e_{\max}=0.2$, and $d_m=0.43$. By adjusting the input bounds $u_m$ between 1.119 and 2.0, we proportionally modify $a_{max}$ within the range of 0.0238 to 0.9048 to meet the conditions of \eqref{eq:amax}.}

\subsubsection{\sh{Prev-CBFs\cite{pati2023preview}}}
\sh{We further compared between the proposed LPrev-CBF with $T_p=10 \,ms$ and our prior Prev-CBF in \cite{pati2023preview}, where the preview horizon $T_p$ is unconstrained and unlimited, implying \yy{full} 
knowledge of \shn{previewable} disturbances throughout the entire horizon.} 


\sh{Figure \ref{fig:1} shows the simulated output and input trajectories under various conditions: without a safety controller, with the standard CBF in \cite{ames2016control}, with Prev-CBF in \cite{pati2023preview}, and with LPrev-CBF, for $u_{m}=1.119$ and $u_{m}=1.8$. Without safety measures, the safety constraint (depicted by black dashed lines) is breached. By contrast, the standard CBF, Prev-CBF, and LPrev-CBF ensure safety. Notably, with a smaller $u_m$, the standard CBF deviates significantly from the  trajectory without safety, while Prev-CBF and LPrev-CBF remain closer. This distinction is also evident in input trajectories. Moreover, the standard CBF intervenes earlier and more aggressively, while Prev-CBF and LPrev-CBF operates primarily near the safety limits. For $u_m=1.8$, interventions from all methods are minimal, although the standard CBF still intervenes sooner than Prev-CBF and LPrev-CBF.

\begin{figure}[t]
\centering
\includegraphics[width=0.5\textwidth,trim=110mm 77cm 10mm 0mm,clip]{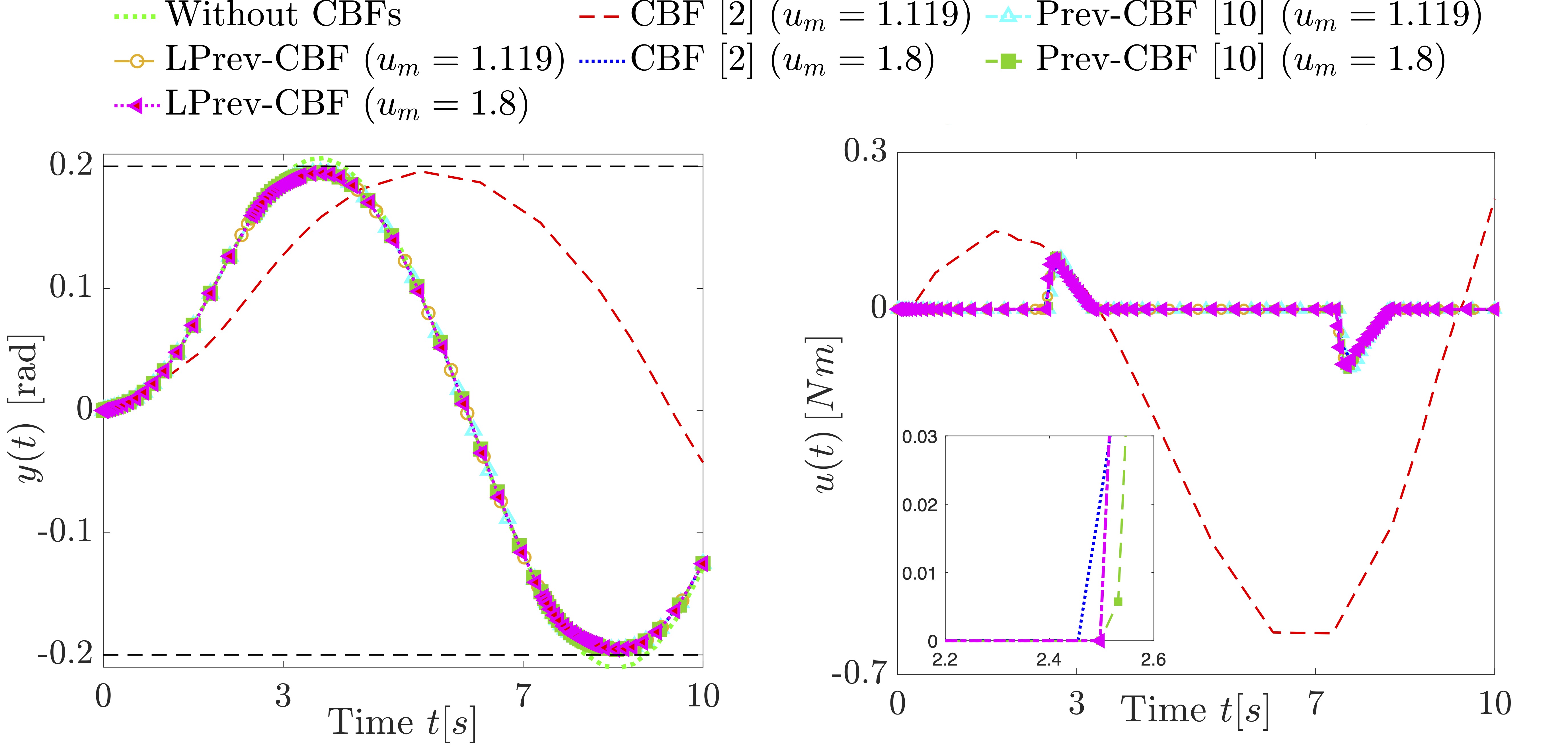}
\includegraphics[width=0.5\textwidth,trim=0mm 0mm 5mm 150mm,clip]{images/figure_revised_1.jpg}
\caption{Angular error (\textbf{left}) and input (\textbf{right}) trajectories: 
(i) Without CBFs (exceeds black dashed bounds), (ii) standard CBF \cite{ames2016control} with  \sy{$u_m=1.119$, (iii) Prev-CBF \cite{pati2023preview} with $u_m=1.119$, (iv) LPrev CBF with  $u_m=1.119$, (v) standard CBF \cite{ames2016control} with  $u_m=1.8$, (vi) Prev-CBF \cite{pati2023preview} with $u_m=1.8$}, and (vii) LPrev CBF with  $u_m=1.8$. Furthermore, an amplified segment in the (\textbf{right}) plot elucidates the initial intervention disparities for all three conditions when $u_m=1.8$.}
\label{fig:1} \vspace{-0.45cm}
\end{figure}

\begin{figure}[t]
\centering
\includegraphics[scale=0.375,trim=0mm 0mm 5mm 0mm]{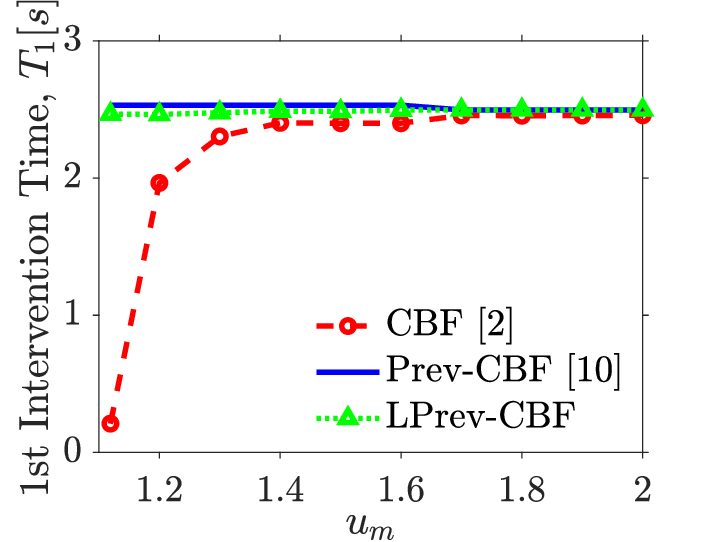}
\includegraphics[scale=0.375,trim=0mm 0mm 10mm 0mm]{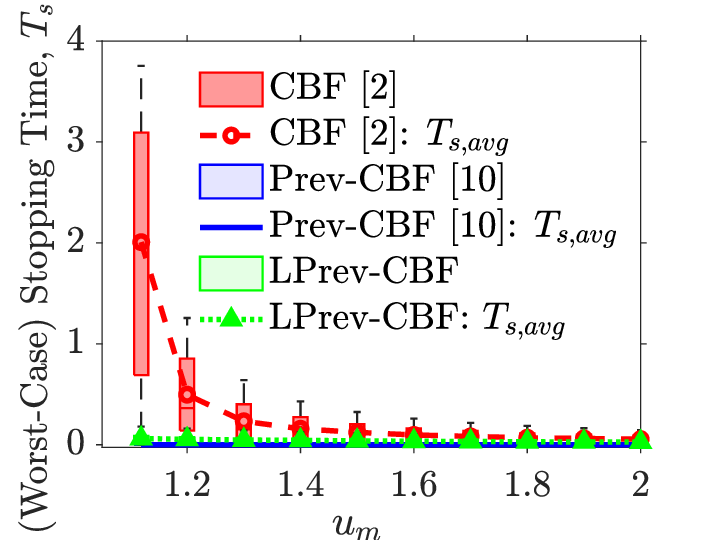}
\caption{Intervention times (\textbf{left}) and stopping times (\textbf{right}, shown as box plots) vary with $u_m$ between 1.119 and 2. The standard CBF \cite{ames2016control} typically acts sooner and has extended stopping times compared to LPrev-CBF. Conversely, Prev-CBF \cite{pati2023preview} shows late intervention and exhibits shorter stopping times than LPrev-CBF. However, these differences become less pronounced as $u_m$ rises}\label{fig:2} \vspace{-0.2cm}
\end{figure}

We also assessed the effect of varying $u_m$ on the initial intervention time $T_1$, marking the first non-zero input instance. A $T_1$ closer to 3.1643 $s$ (the violation time without CBFs) indicates later safety intervention, implying a less conservative safety controller. Additionally, we examined the influence on the (worst-case) stopping time $T_s$. As seen in Figure \ref{fig:2}, LPrev-CBFs intervene later and have lower (worst-case) stopping times than standard CBFs, suggesting their superior utilization of \shn{preview information}, resulting in less conservatism. However, when compared with Prev-CBFs, the LPrev-CBFs intervene earlier and have a little larger (worst-case) stopping times than Prev-CBFs, which is as expected since the \shn{preview} information is more limited with LPrev-CBFs. 
In summary, the results affirm that even with a limited preview horizon, our proposed LPrev-CBF ensures system safety with reduced conservatism.
}

\subsection{Lane Keeping with Road Curvature Preview}\label{AA}
\pt{Next, we consider the lane-keeping example of lateral positioning of a vehicle when limited preview of the road curvature is available. Specifically, we are inspired by the lane-keeping problem in \cite[Section V-B]{ames2016control}, and the relevant predictive state dynamics\footnoteref{note3} can be written as in \eqref{eqn:delay_sys_z} with }
\begin{gather*}
A =
    \begin{bmatrix}
    0 & 1 & v_0 & 0 \\
    0 &-\frac{C_f+C_r}{Mv_0} & 0 & \frac{bC_r-aC_f}{Mv_0}-v_0 \\
    0 & 0 & 0 & 1 \\
    0 &\frac{bC_r-aC_f}{I_zv_0} & 0 & -\frac{a^2C_f+b^2C_r}{I_zv_0} 
    \end{bmatrix},
    B=
    \begin{bmatrix}
    0\\
     \frac{C_f}{M} \\
     0\\
     a\frac{C_f}{I_z}
    \end{bmatrix},\\ B_d=
    \begin{bmatrix}
    0 &
     0 &
     -1 &
     0
    \end{bmatrix}^\top,  C=\begin{bmatrix}
    1 &
     0 &
     0 &
     0
    \end{bmatrix},
\end{gather*}
 \pt{with state $x\triangleq [y, \nu, \psi, r]^\top$, representing lateral velocity $\nu$, lateral displacement $y$, yaw rate $r$ and error yaw angle $\psi$. In this specific setup, the steering angle of the front tire serves as the input $u$ to our system, whereas the desired yaw rate, $r_d=\frac{v_0}{R}$,  is the disturbance with constant longitudinal velocity $v_0$ and (unknown but previewable) road curvature $R$. In this example, we considered a sinusoidal $r_d$ as disturbance. Further, the known signals and system parameters are: Vehicle mass $M=1650\,kg$, distances of rear and front wheels from center of mass $b= 1.59\,m$ and $a=1.11\,m$, respectively, rear and front tire stiffness parameters $C_r=133000\,N/rad$ and $C_f=98800\,N/rad$, respectively, and the vehicle moment of inertial with respective its center of mass $I_z=2315.3\,kgm^2$, taken from \cite[Section V-B]{ames2016control}. 
 
 Additionally, the system has a constant input delay of $T_i=10\,ms$ and a preview of road curvature for a constant preview time $T_p=20\,ms$ along with the initial predictive state $z(0)=[0.5,1.2,0,0]^{\top}$. Next, for stabilizing the vehicle in the center of the lane, we employ a nominal controller $k(z,t)=K(z_{ff}-z)$, with $z_{ff}=[0 \ 0 \ 0 \ r_d]^\top$. Safety here constitutes adhering to the lane boundary constraint $|y|\leq{y_m}$, where $y_m$ is chosen as 0.6 $m$ in this example. Further, the input constraint is $|u|\leq{u_m}$, where we consider three distinct values of $u_m$ for comparison.}
\begin{figure}[t]
\centering
\includegraphics[scale=0.27,trim=10mm 0mm 5mm 0mm]{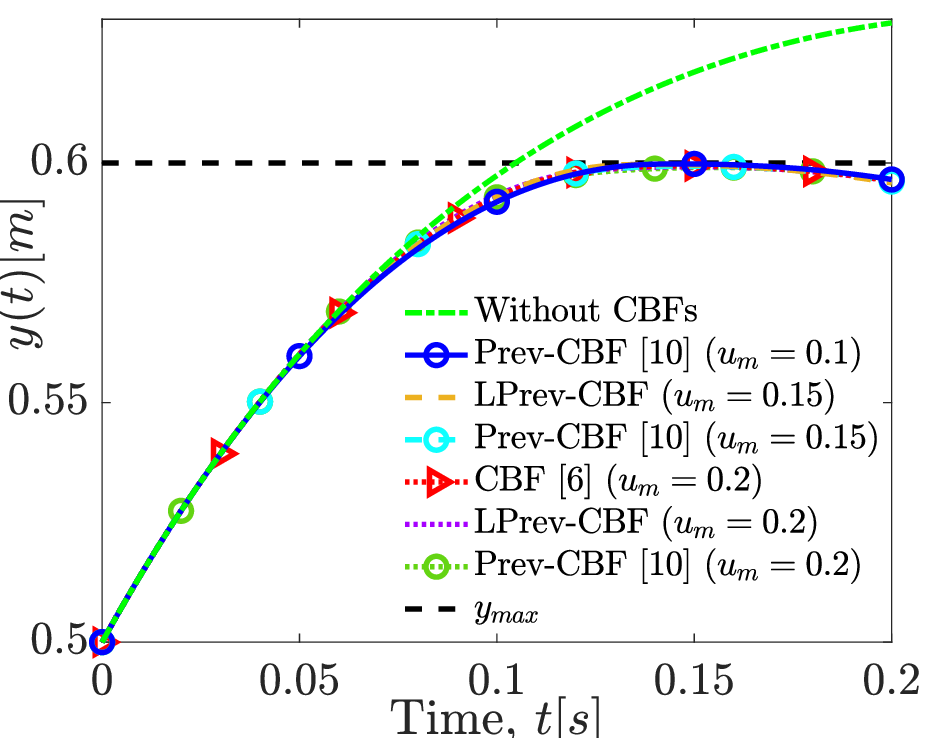} 
\includegraphics[scale=0.27,trim=0mm 0mm 10mm 0mm]{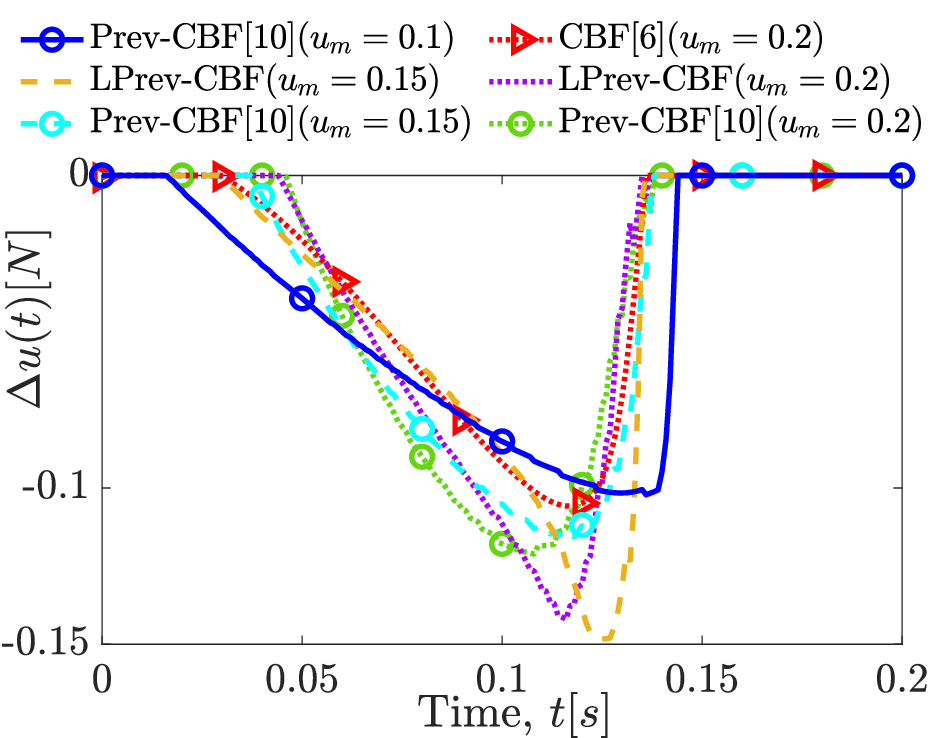}
\caption{Lateral displacement trajectories $y(t)$ (\textbf{left}) and trajectories of CBF intervention of $u(t)$ (\textbf{right}) given by $\Delta u(t)\triangleq u(t)-k(x(t),t)$, where $k(x(t),t)$ is the legacy controller.}\vspace{-0.2cm}
\label{fig:3lk}
\end{figure}
\pt{
\subsubsection{Limited Preview CBFs} 
We apply the proposed LPrev-CBF in \eqref{eq:h} in Lemma \ref{first}  within an optimization-based framework in Proposition \ref{prop:control} to this problem and also compare its performance with our prior design, Prev-CBF \cite{pati2023preview}, with unlimited preview and with the standard CBF \cite{ames2016control} that does not consider preview. 
}
\subsubsection{Preview CBFs}
For comparison, we consider Prev-CBF that we previously proposed in \cite{pati2023preview}, where a sufficiently large or ``infinite" preview horizon $T_p$ is assumed and 
the rest of the simulation parameters for Prev-CBF are kept the same as for the proposed LPrev-CBF.

\subsubsection{Standard CBF}  \pt{Lastly, we also consider the Standard Lane-Keeping CBF proposed in \cite[Section V-B]{ames2016control} given in \eqref{eq:CBFs_ames}, where $y_m$ represents half the lane width and $\dot{y}$ the lateral velocity, calculated from $\dot{y}(t)=\nu + \psi{v_0}$ derived from vehicle dynamics. 
Using the relationship between $\ddot{y}(t)$ and  $u$ derived from the lane-keeping dynamics,  }
\begin{align*}
\ddot{y}=\frac{C_fu - F_0}{M},
\end{align*}
\pt{where $F_0\triangleq C_f\frac{\nu+ar}{v_0}+C_r\frac{\nu-br}{v_0} +Mv_0r_d$ such that $|F_0|\le F_{0,\max}$ with a known $F_{0,\max}$, we can formulate the acceleration limits $a_{\max}$ 
as
\begin{align*}
a_{\max}=\textstyle\frac{1}{M}(C_f u_m - F_{0,\max})
\end{align*}
\pt{and with the control input that satisfies:}
\begin{gather}\label{eq:input-bound-amax2}
   \sy{ u (t) \in \textstyle [\frac{1}{C_f}(-M a_{\max}\hspace{-0.08cm}+\hspace{-0.08cm}F_0(t)),\frac{1}{C_f}(M a_{\max} \hspace{-0.08cm}+\hspace{-0.08cm}F_0(t))].}
\end{gather}

As evident from Figure \ref{fig:3lk} (left), in the absence of any CBFs, the vehicle with just the nominal controller violates the lateral safety condition, whereas the proposed LPrev-CBF under input constraints $u_m \in \{0.2,0.15\}$, 
Prev-CBF  \cite{pati2023preview} under input constraints $u_m\in \{0.2,0.15,0.1\}$ and 
the Standard CBF \cite{ames2016control} under input constraints $u_m=0.2$  ensure that the vehicle stays within its lane.
Moreover, from our simulations, we observe that the approach with the longest preview, i.e., Prev-CBF that has unlimited preview, can remain safe with the least control limit, $u_m\ge 0.09$, while the standard CBF that does not utilize preview requires the highest control authority, $u_m\ge 0.18$, to remain safe; thus, preview is clearly advantageous. On the other hand, the proposed LPrev-CBF provides a middle ground where preview is available but limited, and can maintain safety with a smaller input bound,  $u_m\ge 0.14$, than the standard CBF. In other words, there is some form of partial ordering of the 3 approaches: Prev-CBF $>$ LPrev-CBF $>$ Standard CBF in terms of minimum control authority needed for safety. 

Further, from Figure \ref{fig:3lk} (right), it can be observed that for a given controller (e.g., Prev-CBF with $u_m\in\{0.2,0.15,0.1\}$), 
the 
greater the actuation/control authority is, the lesser the safety controller needs to intervene against the nominal controller, i.e., the intervention time is reduced. 
Similarly, for a fixed  actuation limit (e.g.,  $u_m=0.2$), Prev-CBF interferes less than the proposed LPrev-CBF that in turn interferes less than the standard CBF.}\vspace{0.1cm}

\noindent\sy{\emph{Discussion of Results.}} \pt{To summarize, the examples presented along with the proposed approaches emphasize the value of \shn{preview} information even if  only limited preview is available and also when there is input delay. This provides broader actuation authority range for other control goals, e.g., for maximizing performance, and causes less overall interference against the nominal controller for safety when compared to the cases where \shn{preview} information is not utilized. Further, with the increase in  actuation limit $u_m$, the value or advantages of the \shn{preview} information decreases 
 since the input range becomes large enough to counter effects of relatively smaller 
 worst-case disturbances. The finding of this work confirms findings of \cite{liu2021value} about the value of preview for discrete-time systems and also the findings of \cite{pati2023preview} about the value of preview for continuous-time linear systems.}

\phantom{a}

 
\section{Conclusion}

\yo{In this paper, we introduced a limited
preview control barrier function (LPrev-CBF)  for linear
continuous-time input-delay systems where the preview horizon for the previewable disturbances is limited and fixed, e.g., due to limited sensing ranges. In contrast to the standard CBF
approach that simply considers worst-case disturbances, our approach can leverage \shn{preview} information to reduce conservatism, while avoiding the assumption in 
the Prev-CBF approach that the disturbances are previewable
for an infinite horizon. Further, our LPrev-CBF explicitly takes input constraints/bounds
into consideration and thus, it naturally has recursive feasibility/safety
guarantees.
Future directions include the extensions of limited preview CBFs to consider preview horizons that may be state- or time-dependent as well as the presence of non-previewable uncertainties/disturbances.} 


\bibliographystyle{IEEEtran}
\bibliography{conference_0}

\end{document}